\documentclass[11pt]{article}
\usepackage{graphicx}
\usepackage{subfigure}
\usepackage{float}
\usepackage{amsfonts,amsbsy,amssymb,amsmath,amsthm,amsfonts,mathtools}
\usepackage{verbatim}
\usepackage{parskip}
\usepackage[T1]{fontenc}
\usepackage{multicol}
\usepackage[onehalfspacing]{setspace}
\usepackage{color}
\usepackage[autostyle]{csquotes}
\usepackage{pgf}
\usepackage{hyperref}
\usepackage{thmtools}
\usepackage{caption}
\usepackage{soul}
\usepackage{pdflscape}
\usepackage{relsize}
 
\declaretheoremstyle[%
  spaceabove=-6pt,%
  spacebelow=6pt,%
  headfont=\normalfont\itshape,%
  postheadspace=1em,%
  qed=\qedsymbol%
]{mystyle} 

\hypersetup{
    colorlinks=true,
    linkcolor=blue,
    filecolor=blue,      
    urlcolor=blue,
    citecolor=blue,
    pdftitle={Overleaf Example},
    pdfpagemode=FullScreen,
    }

\newtheorem{thm}{Theorem}[section]
\newtheorem{prop}{Proposition}[section]

\newtheorem{defn}[thm]{Definition}

\newtheorem{remark}[prop]{Remark}
\newtheorem{example}[thm]{Example}

\newcommand{\FF}{{\mathbb{F}}}

\newcommand{\CC}{\mathcal{C}}

\title{ LCPs of Subspace Codes}

\author{Sanjit Bhowmick$^{1}$ 
\footnote{
$^{1}$Department of Electronics and Electrical Engineering, Indian Institute of Technology Guwahati,\\
 Assam, 781039, India. Email: sanjitbhowmick@rnd.iitg.ac.in \\
  }
}

\begin{document}
\maketitle
\begin{abstract} 
A subspace code is a nonempty collection of subspaces of the vector space $\mathbb{F}_q^{n}$.
A pair of linear codes is called a linear complementary pair (in short LCP) of codes if their intersection is trivial and the sum of their dimensions equals the dimension of the ambient space. 
In this paper, we introduce the concept of LCPs of subspace codes. We first provide a characterization of subspace codes that form an LCP. Furthermore, we present a sufficient condition for the existence of an LCP of subspace codes based on a complement function on a subspace code. In addition, we give several constructions of LCPs for subspace codes using various techniques and provide an application to insertion error correction.
\end{abstract}


\noindent\textbf{2020 AMS Classification Code: 12E20, 94B05, 94B60} 

\section{Introduction}\label{sec:intr}
Random network coding greatly improves how data moves through a network, making it faster and more reliable. Its main advantage is that the intermediate nodes can mix data packets together. However, this also creates a significant challenge, because errors do not remain isolated but instead spread through the network in a complex algebraic way. As a result,  traditional error-correcting codes, which are designed for fixed symbol sequences, do not perform well in this context. To address this problem, K{\"o}tter and Kschischang \cite{KR08} introduced a new concept called subspace codes. 
 Instead of treating messages as strings of symbols, each message is represented as a subspace of a vector space. The distance between two subspaces is defined in terms of the difference between the dimension of their sum and the dimension of their intersection. 
In network coding, a message is sent as a subspace $V$. During transmission, some parts of the message might be lost (erasures), or additional unwanted parts might be added (errors). A subspace code can fix up to $t$ errors and $\rho$ erasures, as long as $2(t + \rho) \leq d$, where $d$ is the code’s distance.

On the other hand, the notion of Linear Complementary Dual (LCD) codes was first introduced by Massey in 1992 \cite{Mas92}. He explored their basic algebraic properties, gave concrete examples, and showed that they could be optimal for certain communication channels, such as the two-user binary adder channel. For many years, LCD codes were mostly of theoretical interest. Their practical importance became clear when Bringer et al. \cite{BCC14} showed that they could help design secure cryptographic systems, particularly protecting against side-channel and fault injection attacks. As a result of this finding, researchers extensively explored the structure and construction of LCD codes. A major breakthrough came when Carlet and Guilley \cite{CG16} adapted classical linear code constructions to create LCD codes effectively. After that, many researchers studied these codes and found ways to use them in applications \cite{BDM26, CG16, CMTQ18, CMT18}.
The concept of linear complementary pairs (LCPs) of codes over finite fields was first introduced by Ngo et al. \cite{NBD15}, who also proposed a direct sum construction and studied its security against fault injection and side-channel attacks. Building on this idea, Carlet et al. \cite{CG18} investigated LCPs of constacyclic codes and showed that for an LCP $(C, D)$, the codes $C$ and $D^\perp$ are monomially equivalent. They further observed that this relationship extends to a particular class of quasi-cyclic codes, namely $2D$-cyclic codes. More recently, Bhowmick, Dalai, and Mesnager \cite{BDM23} extended the study of LCPs to algebraic geometry (AG) codes, presenting new constructions derived from algebraic curves, including elliptic curves. These developments highlight the growing interest in LCPs and their potential applications in both coding theory and cryptography. In \cite{BD25}, additive complementary pairs of codes were recently introduced, and their application to the two-user binary adder channel was also presented.

Crnkovi\'c and \v{S}vob \cite{CS25} introduced LCD subspace codes and demonstrated that these codes can be constructed using certain structural partitions of mathematical objects, specifically association schemes. By analyzing the properties of these schemes, they identified systematic methods to generate LCD subspace codes. More recently, Crnkovi\'c, Ishizuka, Kharaghani, Suda, and \v{S}vob \cite{CIKSS24} showed how to make both self-orthogonal and LCD subspace codes using several combinatorial designs, such as weighing matrices and linked symmetric designs, along with special partitions of these designs.
In a recent study, Liu et al. \cite{LHL25} developed a unified approach for analyzing $s$-Galois LCD subspace codes over finite fields. Their main result provides a necessary and sufficient condition for a subspace code to possess the $s$-Galois LCD property, thereby extending the foundational work of Crnković and Švob \cite{CS25}. In addition, the authors introduced three new construction techniques for building such codes. 

The contributions of this paper are summarized as follows:
\begin{itemize}
\item We introduce a general framework that relates the LCPs of subspace codes to the subspace distance between two subspaces.
\item Under suitable conditions, we establish an equivalence between the LCP of a subspace code and that of its dual code.
\item We derive a necessary and sufficient condition for the existence of an LCP of subspace codes in terms of their generator matrices.
\item We propose several constructions for obtaining LCPs of subspace codes.
\end{itemize}

This paper is organized as follows. In Section~\ref{pre-3}, we review the relevant background on subspace codes and linear codes, along with several necessary results concerning these codes. In Section~\ref{lcp-4}, we present a relation between the existence of an LCP for subspace codes and the subspace distance between two codes. Furthermore, we establish that the pair $\{\mathcal{C}, \mathcal{D}\}$ is an LCP of subspace codes if and only if  $\{\mathcal{C}^\perp, \mathcal{D}^\perp\}$ is LCP, provided that  for all $C_i\in \mathcal{C}$ and $D_j\in \mathcal{D}$, $\dim(C_i)+\dim(D_j)=n$, and vice-versa, where $\mathcal{C}^\perp=\{C_i^\perp~|~C_i\in\mathcal{C}\}$. Apart from this, we obtain two necessary and sufficient conditions for subspace codes to be LCP. In Section~\ref{m-5}, we show that $\{\mathbb{F}_q^{n\times m}(\mathcal{C}), \mathbb{F}_q^{n\times m}(\mathcal{D})\}$ is an LCP of subspace code if and only if $\{\mathcal{C}, \mathcal{D}\}$ is an LCP of subspace codes, where $\mathbb{F}_q^{n\times m}(\mathcal{C})=\{\mathbb{F}_q^{n\times m}(U)~|~U\in\mathcal{C}\}$. In addition, Section~\ref{con-6} provides constructions of LCP subspace codes using the $[u|u+v]$-construction, the $[u+v|\lambda u-\lambda v]$-construction (with $\lambda\in\mathbb{F}_q\setminus\{0\}$), and the $k$-spread of $\mathbb{F}_q^n$. Further, we
present an application to insertion error correction in Section~\ref{app-6}. Finally, the paper concludes in Section~\ref{con-7}.
\section{Preliminaries}\label{pre-3}

Let $ \mathbb{F}_q $ denote the finite field with $ q $ elements, where $ q $ is a prime power. Let $C$ be a $q$-ary linear code of dimension $k$ over the finite field $\mathbb{F}_{q}$. In fact, the code $C$ is a $k$-dimensional subspace of $\mathbb{F}_{q}^{n}$, and the elements of $C$ are called codewords. For two vectors $x, y \in \mathbb{F}_{q}^{n}$, the Hamming distance between them is defined as
$d(x,y) = |\{\, i : x_i \neq y_i \,\}|.$
The weight of a codeword $x$ is the number of its nonzero components, that is,
$w(x) = d(x,0) = |\{\, i : x_i \neq 0 \,\}|.$
The minimum distance of the code $C$ is given by
$d = \min\{\, d(x,y) : x, y \in C,\, x \neq y \,\},$
and for a linear code this is equivalent to
$d = \min\{\, w(x) : x \in C,\, x \neq 0 \,\}.$
A code with length $n$, dimension $k$, and minimum distance $d$ is denoted by $[n, k, d]_q$.
The dual code of $C$, denoted $C^\perp$, is defined as
\[
C^\perp = \{\, v \in \mathbb{F}_{q}^{n} ~|~ \langle v, c \rangle = 0 \text{ for all } c \in C \,\},
\]
where $\langle \cdot, \cdot \rangle$ is the Euclidean inner product. A linear code $C$ is called an LCD (Linear Complementary Dual) code if
$C \cap C^\perp = \{0\}.$ More generally, for any two linear codes $C$ and $D$ of the same length $n$, the pair $\{C, D\}$ is called an LCP (linear complementary pair) of codes if $C\oplus D=\mathbb{F}_q^n$. First, we recall a useful proposition.
\begin{prop}\cite[Theorem~2.1]{Guenda2019}\label{p-3}
For $i=1,2$, let $C_i$ be a linear code with a generator matrix $G_i$ and a parity check matrix $H_i$. If $C_1\cap C_2=\{ 0\}$, then $G_1H_2^\top$ and $G_2H_1^\top$ both are right-invertible.  
\end{prop} The family of all linear subspaces of $ \mathbb{F}_q^n $ is called the projective space of order $ n $ and is denoted by $ \mathcal{P}_q(n) $.

For any $ U, V \in \mathcal{P}_q(n) $, we define their sum as
$$
U + V = \{ u + v \mid u \in U, v \in V \}.
$$
The sum is the smallest subspace of $\mathbb{F}_q^n$ that contains both $U$ and $V$. It also satisfies the well-known dimension formula
$$
 \dim(U + V) = \dim(U) + \dim(V) - \dim(U \cap V).   
$$
A subspace code of length $ n $ is any non-trivial collection $\mathcal{C} \subseteq \mathcal{P}_q(n)$ containing at least two distinct subspaces. For any two subspaces $U, V\in\mathcal{P}_q(n)$, the subspace distance is defined as \begin{equation}\label{eq-1.1}
 d_s(U , V) = \dim(U+V) - \dim(U \cap V).
\end{equation} The minimum subspace distance of a subspace code $\mathcal{C} \subseteq \mathcal{P}_q(n)$ is then given by
$$
d_s(\mathcal{C})=\min\{d_s(U,V)~|~ U,V \in \mathcal{C}, U\neq V\}.
$$ 
If every subspace in $\mathcal{C}$ has dimension $ k $ for some $1\leq k \leq n$, then $ \mathcal{C} $ is referred to as a constant-dimension subspace code.

For any subspace code $\mathcal{C}$, we can associate another subspace code $\mathcal{C}^{\perp}$, called its complementary code, defined by
\[
\mathcal{C}^{\perp}=\{U^{\perp}\mid U\in \mathcal{C}\},
\]
where $U^{\perp}$ denotes the orthogonal complement of $U$.
Because the subspace distance satisfies
\[
d_s(U,V)=d_s(U^{\perp},V^{\perp}) \quad \text{for all } U,V\in \mathcal{P}_q(n),
\]
the minimum distance is preserved, that is, $d(C)=d(C^{\perp})$.

Furthermore, note that
\[
(U^{\perp})^{\perp}=U,\qquad (U+V)^{\perp}=U^{\perp}\cap V^{\perp},\qquad (U\cap V)^{\perp}=U^{\perp}+V^{\perp}.
\]
Next, we present a definition as given in \cite{BEV13, MB25}. 
\begin{defn}\cite[Definition~1]{BEV13}\label{de-2.2}
 Suppose $\mathcal{C} \subseteq \mathcal{P}_q(n)$, and let $\mathcal{C}_k$ denote the subset of $\mathcal{C}$ consisting of all $k$-dimensional subspaces. A map
$f : \mathcal{C} \to \mathcal{C}$
is called a complement on $\mathcal{C}$ if it satisfies  
\begin{enumerate}
\item[a)] $X \cap f(X) = {0}$ and $X + f(X) = \mathbb{F}_q^{n}$, i.e.,  $X \oplus f(X)=\mathbb{F}_q^{n}$.
\item[b)] $f$ provides a bijection between $\mathcal{C}_k$ and $\mathcal{C}_{n-k}$.
\item[c)] $f$ is an involution, i.e., $f(f(X))=X$ for all $X\in \mathcal{C}$.
\item[d)] $f$ preserves the subspace distance, $d_s(f(X), f(Y)) = d_s(X,Y)$ for all $X,Y\in \mathcal{C}$.
\end{enumerate}
\end{defn}

\section{LCPs of subspace codes}\label{lcp-4}
In this section, we explore LCPs arising from two subspace codes and provide a generalization of the main theorem in \cite{CS25} to this broader setting. This extension offers new insights into the structure and properties of LCPs within the framework of subspace codes.
We now define LCPs for subspace codes. Let $\mathcal{C}=\{C_i\}_{i\in I}$ be a subspace code indexed by a set $I$. A complementary family is a family $\mathcal{D}=\{D_i\}_{i\in I}$ such that $\dim(D_i)=\dim(C_i^\perp)$ for all $i$. (Note that $\mathcal{D}$ is not uniquely determined; a canonical choice is $\mathcal{D}=\mathcal{C}^\perp$.)

\begin{defn}\label{d-3.1}
Let $\mathcal{C}=\{C_i\}_{i\in I}$ and $\mathcal{D}=\{D_i\}_{i\in I}$ be families of subsets of $\mathcal{P}_q(n)$ with $\dim(D_i)=\dim(C_i^\perp)$. If
\[
C_i\cap D_j=\{0\}\quad\text{for all }i,j\in I,
\]
then $\{\mathcal{C},\mathcal{D}\}$ is called a linear complementary pair (LCP) of subspace codes.
\end{defn}
Clearly, $\mathcal{D}$ is a subset of the projective space $\mathcal{P}_q(n)$.
If, in particular, we take $\mathcal{D} = \mathcal{C}^\perp$, then the pair $\{\mathcal{C}, \mathcal{C}^\perp\}$ forms a linear complementary dual (LCD) subspace code.
\begin{prop}\label{p-3.1}
 Let $\mathcal{C},~\mathcal{D}\subseteq \mathcal{P}_q(n)$ be two subspace codes. Then, the pair $\{\mathcal{C}, \mathcal{D}\}$ is an LCP of subspace codes if and only if $d_s(C_i, D_j)=\dim(C_i)+\dim(D_j),$ for each $C_i\in\CC$, $D_j\in\mathcal{D}$.   
\end{prop}
\begin{proof}
 By Definition~\ref{d-3.1}, we know that $\{\mathcal{C},\mathcal{D}\}$ is an LCP of subspace codes if and only if $C_i\cap D_j=\{0\}$, for all $C_i\in \mathcal{C}$, $D_j\in \mathcal{D}$. Next, by \eqref{eq-1.1}, we have $
d_s(C_i , D_j) = \dim(C_i+D_j) - \dim(C_i \cap D_j)=\dim(C_i)+\dim(D_j)-2\dim(C_i\cap D_j). 
$ Hence, $\{\mathcal{C}, \mathcal{D}\}$ is LCP if and only if $d_s(C_i, D_j)=\dim(C_i)+\dim(D_j),$ for each $C_i\in\CC$, $D_j\in\mathcal{D}$.
\end{proof}
\begin{thm}\label{th-3.2}
  Let $\mathcal{C},~\mathcal{D}\subseteq \mathcal{P}_q(n)$ be two subspace codes. Then, the following statements are equivalent.
  \begin{enumerate}
 \item[1)] $\{\mathcal{C}, \mathcal{D}\}$ is an LCP of subspace codes if and only if  $\{\mathcal{C}^\perp, \mathcal{D}^\perp\}$ is LCP.
 \item[2)] For all $C_i\in \mathcal{C}$ and $D_j\in \mathcal{D}$, $\dim(C_i)+\dim(D_j)=n$.
  \end{enumerate} 
\end{thm}
\begin{proof}
$1) \Rightarrow 2)$: Assume $\{\mathcal{C}, \mathcal{D}\}$ is an LCP of subspace codes if and only if  $\{\mathcal{C}^\perp, \mathcal{D}^\perp\}$ is LCP. 

By Proposition \ref{p-3.1}, for every $C_i\in\mathcal C$ and $D_j\in\mathcal D$,
\[
d_s(C_i,D_j)=\dim C_i+\dim D_j
\quad\Longleftrightarrow\quad
d_s(C_i^\perp,D_j^\perp)=\dim C_i^\perp+\dim D_j^\perp.
\] 
But the subspace distance is invariant under orthogonal complement
\[
d_s(C_i,D_j)=d_s(C_i^\perp,D_j^\perp).
\]
Therefore,
\[
\dim C_i+\dim D_j=\dim C_i^\perp+\dim D_j^\perp.
\]
Using $\dim C^\perp=n-\dim C$, this becomes
\[
\dim C_i+\dim D_j
= (n-\dim C_i)+(n-\dim D_j)
\quad\Longrightarrow\quad
\dim C_i+\dim D_j = n .
\]
 $2)\Rightarrow 1)$: 
Assume that $\dim(C_i) + \dim(D_j) = n$ for all $C_i \in \mathcal{C} \text{ and } D_j \in \mathcal{D}$. Then, clearly, $\dim(C_i^\perp) + \dim(D_j^\perp) = n$. By Proposition~\ref{p-3.1}, $\{\mathcal{C}, \mathcal{D}\}$ forms LCP if and only if $d_s(C_i, D_j) = \dim(C_i) + \dim(D_j) $ for all $C_i \in \mathcal{C}$ and $D_j \in \mathcal{D}$.

Moreover, since
\[
d_s(C_i^\perp, D_j^\perp) = d_s(C_i, D_j) = n = \dim(C_i^\perp) + \dim(D_j^\perp),
\]
it follows that $\{\mathcal{C}^\perp, \mathcal{D}^\perp\}$ is also an LCP of subspace codes.

Conversely, if $\{\mathcal{C}^\perp, \mathcal{D}^\perp\}$ is a LCP, then $d_s(C_i^\perp, D_j^\perp) = \dim(C_i^\perp) + \dim(D_j^\perp)$, for all $C_i \in \mathcal{C} \text{ and } D_j \in \mathcal{D}$. Note that
\[ d_s(C_i, D_j)=
d_s(C_i^\perp, D_j^\perp) = \dim(C_i^\perp) + \dim(D_j^\perp) = n = \dim(C_i) + \dim(D_j),
\]
which immediately implies that $\{\mathcal{C}, \mathcal{D}\}$ is LCP as well.
\end{proof}
In the following theorem, we show a characterization of LCP of subspace codes that is analogous to the result of LCD subspace codes in \cite{CS25}.
\begin{thm}\label{th-3.3}
 Let $\mathcal{C},~\mathcal{D}\subseteq \mathcal{P}_q(n)$ be two subspace codes. For each $C_i\in\mathcal{C}$ and $D_j\in\mathcal{D}$, let $G_{C_i}$ and $G_{D_j}$ denote generator matrices, and let $H_{C_i}$ and $H_{D_j}$ denote parity check matrices, respectively. Then the pair $\{\mathcal{C}, \mathcal{D}\}$ is an LCP of subspace codes if and only if $G_{C_i}H_{D_j}^\top$ is right-invertible or $G_{D_j}H_{C_i}^\top$ is right-invertible, for all $C_i\in \mathcal{C}$, $D_j\in\mathcal{D}$.
\end{thm}
\begin{proof}
To prove the sufficient condition, let 
 $x\in C_i\cap D_j$. Then there exist vectors 
 $\alpha\in \FF_q^{k_i}$ and $\beta\in \FF_q^{k_j}$ such that $$x=\alpha G_{C_i} \text{ and } x=\beta G_{D_j},$$ where $k_i$ and $k_j$ are the dimensions of $C_i$ and $D_j$, respectively. This gives $\alpha G_{C_i}H_{D_j}^\top=0$. Since $G_{C_i}H_{D_j}^\top$ is right-invertible, it follows that $\alpha=0$, and hence $x=0$. Therefore, $C_i\cap D_j=\{0\}$ for all $C_i\in \mathcal{C}$, $D_j\in\mathcal{D}$. Thus, by Definition~\ref{d-3.1},  the pair $\{\mathcal{C}, \mathcal{D}\}$ forms an LCP of subspace codes. 

To prove the necessary condition, suppose that $G_{C_i} H_{D_j}^\top$ is not right-invertible. Then there exists a nonzero vector $\delta \in \mathbb{F}_q^{k_i}$ such that
\[
\delta G_{C_i} H_{D_j}^\top = 0,
\]
where $k_i = \dim(C_i)$.
Since $\delta G_{C_i} \in C_i$ and $\delta G_{C_i} H_{D_j}^\top = 0$, it follows that $\delta G_{C_i} \in D_j$. Therefore,
\[
0 \neq \delta G_{C_i} \in C_i \cap D_j,
\]
contradicting the assumption that $C_i \cap D_j = \{0\}$.
This concludes the proof.
\end{proof}
Now, we give a second necessary and sufficient condition for LCP of subspace codes. 
\begin{thm}
  Let $\mathcal{C},~\mathcal{D}\subseteq \mathcal{P}_q(n)$ be two subspace codes. For each $C_i\in\mathcal{C}$ and $D_j\in\mathcal{D}$, let $G_{C_i}$ and $G_{D_j}$ denote generator matrices, respectively. Assume that $\mathcal{G}_{\left(C_i,D_j\right)}=\left( \begin{array}{cc}
       G_{C_i}  \\
       G_{D_j}
  \end{array}\right)$ is an $n\times n$ matrix. Then $\{\mathcal{C},\mathcal{D}\}$ is an LCP of subspace codes if and only if $\mathcal{G}_{\left(C_i,D_j\right)}$ is invertible, for all $C_i\in\mathcal{C}$ and $D_j\in\mathcal{D}$.
\end{thm}
\begin{proof}
 Observe that
\[
\mathcal{G}_{\left(C_i,D_j\right)}=\begin{pmatrix} G_{C_i} \\ G_{D_j} \end{pmatrix}
\]
is an $n\times n$ matrix. Hence $\dim(C_i)+\dim(D_j)=n$ for every $C_i\in\mathcal{C}$ and $D_j\in\mathcal{D}$.

To prove the sufficient condition, take any $x\in C_i^\perp\cap D_j^\perp$.
Then
\[
G_{C_i}x^\top=0 \quad\text{and}\quad G_{D_j}x^\top=0,
\]
which together imply
\[
\mathcal{G}_{\left(C_i,D_j\right)}x^\top=0.
\]
Since $\mathcal{G}_{\left(C_i,D_j\right)}$ is invertible, the only solution is $x=0$. Thus
$C_i\cap D_j={0}$.
Consequently, by applying Theorem~\ref{th-3.2}, we conclude that $\{\mathcal{C}^\perp,\mathcal{D}^\perp\}$ is LCP as well.

To prove the necessary condition, we show that the matrix $\mathcal{G}_{\left(C_i,D_j\right)}$ must be nonsingular. Assume, to the contrary, that $\mathcal{G}_{\left(C_i,D_j\right)}$ is singular. Then there exists a nonzero vector $y \in \mathbb{F}_q^n$ such that
\[
\mathcal{G}_{\left(C_i,D_j\right)}y^\top = 0,
\qquad\text{i.e.,}\qquad
\begin{pmatrix}
G_{C_i} \\
G_{D_j}
\end{pmatrix}
y^\top = 0.
\]
Consequently,
\[
G_{C_i}y^\top = 0
\quad\text{and}\quad
G_{D_j}y^\top = 0,
\]
which implies that
\[
0 \neq y \in C_i^\perp \cap D_j^\perp.
\]
Thus $\{ \mathcal{C}^\perp, \mathcal{D}^\perp \}$ fails to form an LCP of subspace codes. By Theorem~\ref{th-3.2}, this means that $\{ \mathcal{C}, \mathcal{D} \}$ cannot be an LCP of subspace codes, contradicting our assumption. This completes the proof.
\end{proof}

For a subspace code $\mathcal{C} \subseteq \mathcal{P}_q(n)$ and an integer $k$ with $1 \le k \le n-1$, define
\[
\mathcal{C}_k \coloneqq \{ X \in \mathcal{C} \mid \dim(X) = k \},
\]
that is, $\mathcal{C}_k$ denotes the collection of all $k$-dimensional subspaces in $\mathcal{C}$.

\begin{thm}\label{th-3.5}
Let $\mathcal{C}\subseteq\mathcal{P}_q(n)$ be a subspace code and let $f$ be a complement function on $\mathcal{C}$ (see Definition~\ref{de-2.2}). For a fixed integer $k$ with $1\le k\le n-1$, denote $\mathcal{C}_k=\{X\in\mathcal{C}\mid \dim X=k\}$. Suppose that
\[
X+f(Y)=\mathbb{F}_q^n \qquad\text{for all } X,Y\in\mathcal{C}_k.
\]
Then the pair $\{\mathcal{C}_k,\; f(\mathcal{C}_k)\}$ is an LCP of subspace codes.
\end{thm}

\begin{proof}
Since $f$ is a complement function, it maps $\mathcal{C}_k$ bijectively onto $\mathcal{C}_{n-k}$ and satisfies $Y\oplus f(Y)=\mathbb{F}_q^n$ for every $Y\in\mathcal{C}$. In particular, $\dim f(Y)=n-k$.

Take arbitrary $X,Y\in\mathcal{C}_k$. By hypothesis, $X+f(Y)=\mathbb{F}_q^n$. The dimension formula gives
\[
\dim\bigl(X+f(Y)\bigr)=\dim X+\dim f(Y)-\dim\bigl(X\cap f(Y)\bigr).
\]
Substituting $\dim X=k$, $\dim f(Y)=n-k$, and $\dim(X+f(Y))=n$, we obtain
\[
n = k + (n-k) - \dim\bigl(X\cap f(Y)\bigr),
\]
hence $\dim\bigl(X\cap f(Y)\bigr)=0$, i.e., $X\cap f(Y)=\{0\}$. This holds for every $X,Y\in\mathcal{C}_k$.

Thus $\{\mathcal{C}_k,\; f(\mathcal{C}_k)\}$ satisfies the condition of Definition~\ref{d-3.1} (with the obvious indexing $C_i\leftrightarrow X$, $D_i\leftrightarrow f(X)$). Therefore it is an LCP of subspace codes. 
\end{proof}
\begin{remark}
Let $\mathcal{C}=\{C_1,\ldots,C_s\}$ be a collection of $k$-dimensional subspaces of $\mathbb{F}_q^n$, and let $\mathcal{D}=\{D_1,\ldots,D_s\}$ be a collection of $(n-k)$-dimensional subspaces of $\mathbb{F}_q^n$. Suppose that $\{\mathcal{C},\mathcal{D}\}$ forms an LCP of subspace codes. Define a map $f : \mathcal{C}\cup\mathcal{D} \longrightarrow \mathcal{C}\cup\mathcal{D}$ by
\[
f(x) = \left\{\begin{array}{ll}
D_i & \textrm{if }x=C_i\in\mathcal{C};\\
C_i & \textrm{if }x=D_i\in\mathcal{D}; \quad (1\le i\le s)\\
0 & \textrm{otherwise.}
\end{array}\right.
\]
The map $f$ satisfies the following properties:
\begin{enumerate}
    \item[a)]  $X \cap f(X)=\{0\}, \qquad X+f(X)=\mathbb{F}_q^n,$
    i.e., $X\oplus f(X)=\mathbb{F}_q^n$.
    
    \item[b)]  The map $f$ gives a bijection between $\mathcal{C}$ and $\mathcal{D}$.
    
    \item[c)] $f(f(X))=X$ for all $X\in\mathcal{C}\cup\mathcal{D}$.
    
    \item[d)]
    $d_s(f(X),f(Y)) = d_s(X,Y) \qquad \text{for all } X,Y\in\mathcal{C}\cup\mathcal{D}.$
\end{enumerate}
Thus, $f$ forms a complement function on $\mathcal{C}\cup\mathcal{D}$.
\end{remark}
\section{LCP subspace codes derived from matrix codes over $\mathbb{F}_q$}\label{m-5}
For any matrix $A \in \mathbb{F}_q^{n \times n}$, denote its $i $-th column by
$A_i = (a_{1i}, a_{2i}, \ldots, a_{ni})^T .$
The vector space spanned by the transpose of the columns of a matrix $A \in \mathbb{F}_q^{n \times m} $ can be written as
$\operatorname{colsp}(A) \subseteq \mathbb{F}_q^{n}.$
Furthermore, we use the standard $\mathbb{F}_q$-vector space isomorphism
$\mathbb{F}_q^{n \times m} \cong \mathbb{F}_{q^m}^{n}$.

\begin{defn}
Let $ U \subseteq \mathbb{F}_q^{n} $ be a linear code. Consider the set of all matrices
\[
\mathbb{F}_q^{n \times m}(U):=\{A \in \mathbb{F}_q^{n \times m} ~ \mid ~ \operatorname{colsp}(A) \subseteq U \}.
\]
This forms a linear code in $ \mathbb{F}_q^{n \times m} $, called the matrix code induced by $U $.  
\end{defn}
\begin{thm}
   Let $\mathcal{C}~\mathcal{D}\subseteq \mathcal{P}_q(n)$ be two subspace codes and define
\[
\mathbb{F}_q^{n \times m}(\mathcal{C}) = \{ \mathbb{F}_q^{n \times m}(U) \mid U \in \mathcal{C} \} \text{ and  } \mathbb{F}_q^{n \times m}(\mathcal{D}) = \{ \mathbb{F}_q^{n \times m}(V) \mid V \in \mathcal{D} \}.
\]
Then $\{\mathbb{F}_q^{n \times m}(\mathcal{C}), \mathbb{F}_q^{n \times m}(\mathcal{D}) \}$  is an  LCP of subspace codes if and only if  $\{\mathcal{C},\mathcal{D}\}$ is an  LCP of subspace codes.
 \end{thm}
\begin{proof}
We first establish the identity
\[
\mathbb{F}_{q}^{n\times m}(U)\cap \mathbb{F}_{q}^{n\times m}(V)
=\mathbb{F}_{q}^{n\times m}(U\cap V).
\]
Take any matrix
\[
A\in \mathbb{F}_{q}^{n\times m}(U)\cap \mathbb{F}_{q}^{n\times m}(V).
\]
Thus, every column of $A$ lies both in $U$ and in $V$, which implies
\[
\operatorname{colsp}(A)\subseteq U\cap V.
\]
Therefore,
\begin{equation}\label{e-1q}
   \mathbb{F}_{q}^{n\times m}(U)\cap \mathbb{F}_{q}^{n\times m}(V)
\subseteq
\mathbb{F}_{q}^{n\times m}(U\cap V) 
\end{equation}
In fact, on the other hand, let
\[
B\in \mathbb{F}_{q}^{n\times m}(U\cap V).
\]
Then
$\operatorname{colsp}(B)\subseteq U\cap V$,
so in particular
\[
\operatorname{colsp}(B)\subseteq U
\quad\text{and}\quad
\operatorname{colsp}(B)\subseteq V.
\]
Thus,
\begin{equation}\label{e-2q}
  \mathbb{F}_{q}^{n\times m}(U\cap V)
\subseteq
\mathbb{F}_{q}^{n\times m}(U)\cap \mathbb{F}_{q}^{n\times m}(V).  
\end{equation}
From \eqref{e-1q} and \eqref{e-2q}, we conclude
\[
\mathbb{F}_{q}^{n\times m}(U)\cap \mathbb{F}_{q}^{n\times m}(V)
=\mathbb{F}_{q}^{n\times m}(U\cap V).
\]
By Definition~\ref{d-3.1}, $\{\mathcal{C}, \mathcal{D}\}$ is an LCP of subspace codes exactly when
\[
U\cap V=\{0\} \qquad \text{for all } U\in \mathcal{C} \text{ and } V\in\mathcal{D}.
\]
The rest of the theorem follows immediately.
\end{proof}
\section{Some constructions of LCPs  of subspace codes}\label{con-6}
\subsection{LCP of subspace codes derived from $[u|u+v]$-construction}
Let $C_1$ and $C_2$ be linear codes over $\FF_q$ with parameters $[n, k_1]_q$ and $[n, k_2]_q$, respectively.
The $[u, u+v]$-construction (also known as the Plotkin sum) formed from $C_1$ and $C_2$, denoted by $\mathcal{P}(C_1, C_2)$, is defined as
\[
\mathcal{P}(C_1, C_2) : = \{ (u, u+v) \mid u \in C_1, v \in C_2 \}.
\]
Let $G_i$ and $H_i$ denote a generator matrix and a parity check matrix of $C_i$, respectively, for $i=1,2$. Then generator and parity check matrices for $\mathcal{P}(C_1, C_2)$ are given by
\[
\mathcal{G}=\begin{pmatrix} 
G_{1} & G_1 \\ 
0 & G_{2} 
\end{pmatrix} \text{ and } \mathcal{H}=\begin{pmatrix} 
H_{1} & 0 \\ 
-H_{2} & H_{2} 
\end{pmatrix}.
\]
\begin{thm}
 Let $\mathcal{C}_i~\mathcal{D}_i\subseteq \mathcal{P}_q(n)$ be a subspace code, for $i=1,2$. Suppose the pairs $\{\mathcal{C}_1, \mathcal{D}_1\}$ and  $\{\mathcal{C}_2, \mathcal{D}_2\}$  are LCPs of subspace codes. Then the pair $\{\{\mathcal{P}(C_1, C_2)~|~C_1\in\mathcal{C}_1, C_2\in\mathcal{C}_2\}, \{\mathcal{P}(D_1, D_2)~|~D_1\in\mathcal{D}_1, D_2\in\mathcal{D}_2\}\}$ forms an LCP of subspace codes.
\end{thm}
\begin{proof}
 To prove this, let $\mathcal{G}_1$ and $\mathcal{H}_1$ be a generator matrix and a parity check matrix for $\mathcal{P}(C_1, C_2)$  and $\mathcal{G}_2$ and $\mathcal{H}_2$ be the corresponding matrices for $\mathcal{P}(D_1, D_2)$, respectively. It is enough to show that $\mathcal{G}_1\mathcal{H}_2^\top$ is invertible. 
 Further, let $G_i$ be a generator matrix of $C_i$ and $H_i$ be a parity check matrix of $D_i$, for $i=1,2$. Therefore, we may write
  \[
\mathcal{G}_1=\begin{pmatrix} 
G_{1} & G_1 \\ 
0 & G_{2} 
\end{pmatrix} \text{ and } \mathcal{H}_2=\begin{pmatrix} 
H_{1} & 0 \\ 
-H_{2} & H_{2} 
\end{pmatrix}.
\]
Now, 
 \[
\mathcal{G}_1 \mathcal{H}_2^\top=\begin{pmatrix} 
G_{1} & G_1 \\ 
0 & G_{2} 
\end{pmatrix}  
\begin{pmatrix} 
H_{1}^\top & -H_{2}^\top \\ 
0 & H_{2}^\top 
\end{pmatrix}=\begin{pmatrix} 
G_1H_{1}^\top & 0 \\ 
0 & G_2H_{2}^\top 
\end{pmatrix}.
\]
Since, $C_1\cap D_1=\{0\}$ and $C_2\cap D_2=\{0\}$, then by applying Proposition~\ref{p-3}, both $G_1H_1^\top$ and $G_2H_2^\top$ are right-invertible.  Then by Theorem~\ref{th-3.3}, the result follows immediately. 
\end{proof}
The $[u+v,,v]$-construction formed from $C_1$ and $C_2$, denoted by $\Tilde{\mathcal{P}}(C_1, C_2)$, defined as
\[
\Tilde{\mathcal{P}}(C_1, C_2) : = \{ (u+v, v) \mid u \in C_1, v \in C_2 \}.
\] 
In fact, $\Tilde{\mathcal{P}}(C_1, C_2)= \mathcal{P}(C_1+C_2, -C_1)$.
Let $G_i$ and $H_i$ denote a generator matrix and a parity check matrix of $C_i$, respectively, for $i=1,2$. Then generator and parity check matrices for $\Tilde{\mathcal{P}}(C_1, C_2)$ are given by
\[
\Tilde{\mathcal{G}}=\begin{pmatrix} 
G_{1} & 0 \\ 
G_2 & G_{2} 
\end{pmatrix} \text{ and } \Tilde{\mathcal{H}}=\begin{pmatrix} 
H_{1} & -H_1 \\ 
0 & H_{2} 
\end{pmatrix}.
\]
\begin{thm}
Let $\mathcal{C}~\mathcal{D}\subseteq \mathcal{P}_q(n)$ be two subspace codes. Suppose the pair $\{\mathcal{C}, \mathcal{D}\}$ forms an LCP of subspace code. Then the pair $\left\{\{\mathcal{P}(C_1, C_2)~|~C_1\in\mathcal{C}, C_2\in\mathcal{D}\}, \{\Tilde{\mathcal{P}}(C_1, C_2)~|~C_1\in\mathcal{C}, C_2\in\mathcal{D}\}\right\}$ is also an LCP of subspace codes.
\end{thm}
\begin{proof}
 Let $\mathcal{G}_1$ and $\mathcal{H}_1$ denote a generator matrix and a parity-check matrix for $\mathcal{P}(C_1, C_2)$, and let $\mathcal{G}_2$ and $\mathcal{H}_2$ be the corresponding matrices for $\widetilde{\mathcal{P}}(C_1, C_2)$.
To establish the claim, it suffices to show that the matrix $\mathcal{G}_1 \mathcal{H}_2^\top$ is invertible.
For $i = 1,2$, let $G_i$ be a generator matrix of $C_i$ and $H_i$ a parity-check matrix of $C_i$. Therefore, we may write
  \[
\mathcal{G}_1=\begin{pmatrix} 
G_{1} & G_1 \\ 
0 & G_{2} 
\end{pmatrix} \text{ and } \mathcal{H}_2=\begin{pmatrix} 
H_{1} & -H_1 \\ 
0 & H_{2} 
\end{pmatrix}.
\]
Now, 
 \[
\mathcal{G}_1 \mathcal{H}_2^\top=\begin{pmatrix} 
G_{1} & G_1 \\ 
0 & G_{2} 
\end{pmatrix}  
\begin{pmatrix} 
H_{1}^\top & 0 \\ 
-H_{1}^\top & H_{2}^\top 
\end{pmatrix}=\begin{pmatrix} 
 0 & G_1H_{2}^\top \\ 
-G_2H_{1}^\top & 0 
\end{pmatrix}.
\]
Since $C_1 \cap C_2 = \{0\}$, by Proposition~\ref{p-3}, we find that the matrices $G_1H_2^\top$ and $H_2G_1^\top$ are invertible.
Thus, $\mathcal{G}_1\mathcal{H}_2^\top$ is right-invertible. Then by Theorem~\ref{th-3.3}, the desired result follows. 
\end{proof}
\subsection{LCP of subspace codes derived from $[u+v|\lambda u-\lambda v]$}
For two linear codes $C_1:=[n, k_1]$ and $C_2:=[n, k_2]$, define a linear code
$$
\mathcal{S}_{\lambda}(C_1, C_2):=\{(u+v,\lambda u-\lambda v)~|~u\in C_1, v\in C_2\}
.$$ 
Codes produced through the $[u+v|\lambda u-\lambda v]$-construction exhibit several notable structural advantages. In particular, when $C_1$ is a cyclic code and $C_2$ is a negacyclic code, then the code $\mathcal{S}_{1}(C_1, C_2)$ itself  becomes a cyclic code. Further details can be found in Theorem 8.1 in \cite{Hu00}.
Denote $\mathcal{S}(C_1, C_2)=\mathcal{S}_{1}(C_1, C_2)$.
 Let $C_1$ be linear code with generator matrix $G_1$ and 
parity-check matrix $H_1$, and let $C_2$ be linear code with
generator matrix $G_2$ and parity-check matrix $H_2$. Thus a generator matrix of $\mathcal{S}_{\lambda}(C_1, C_2)$ is
\[
\mathcal{G} =
\begin{pmatrix}
G_1 & \lambda G_1 \\
G_2 & -\lambda G_2
\end{pmatrix}
\] and a parity-check matrix of $\mathcal{S}_{\lambda}(C_1, C_2)$ is
\[
\mathcal{H} =
\begin{pmatrix}
H_1 & \lambda^{-1} H_1 \\
H_2 & -\lambda^{-1} H_2
\end{pmatrix}.
\]
\begin{thm}
 Let $\mathcal{C}_i~\mathcal{D}_i\subseteq \mathcal{P}_q(n)$ (with $q$ odd) be two subspace codes, for $i=1,2$. Suppose the pairs $\{\mathcal{C}_1, \mathcal{D}_1\}$ and  $\{\mathcal{C}_2, \mathcal{D}_2\}$  are LCPs of subspace codes. Then the pair $\{\{\mathcal{S}_{\lambda}(C_1, C_2)~|~C_1\in\mathcal{C}_1, C_2\in\mathcal{C}_2\}, \{\mathcal{S}_{\lambda}(D_1, D_2)~|~D_1\in\mathcal{D}_1, D_2\in\mathcal{D}_2\}\}$ forms an LCP of subspace codes.
\end{thm}
\begin{proof}
 Let $\mathcal{G}_1$ and $\mathcal{H}_1$ be, respectively, a generator matrix and a parity-check matrix for $\mathcal{S}_\lambda(C_1,C_2)$, and similarly let $\mathcal{G}_2$ and $\mathcal{H}_2$ correspond to $\mathcal{S}_\lambda(D_1,D_2)$.
It suffices to show that $\mathcal{G}_1 \mathcal{H}_2^\top$ is invertible.
Let $G_i$ be a generator matrix of $C_i$ and $H_i$ a parity-check matrix of $D_i$, for $i=1,2$. Therefore, we may write
  \[
\mathcal{G}_1=\begin{pmatrix} 
G_{1} & \lambda G_1 \\ 
G_2 & -\lambda G_{2} 
\end{pmatrix} \text{ and } \mathcal{H}_2=\begin{pmatrix} 
H_{1} & \lambda^{-1}H_1 \\ 
H_{2} & -\lambda^{-1} H_{2} 
\end{pmatrix}.
\]
Now, 
 \[
\mathcal{G}_1 \mathcal{H}_2^\top=\begin{pmatrix} 
G_{1} & \lambda G_1 \\ 
G_2 & -\lambda G_{2}
\end{pmatrix}  
\begin{pmatrix} 
H_{1}^\top & H_2^\top \\ 
\lambda^{-1}H_1^\top & -\lambda^{-1} H_{2}^\top 
\end{pmatrix}=\begin{pmatrix} 
2G_1H_{1}^\top & 0 \\ 
0 & 2G_2H_{2}^\top 
\end{pmatrix}.
\]
Since, $C_i\cap D_i=\{0\}$, then by Proposition~\ref{p-3}, each product $G_iH_i^\top$ is right-invertible. Since $q$ is odd, hence $2G_iH_i^\top$ is right-invertible for $i=1,2$. Therefore, $\mathcal{G}_1\mathcal{H}_2^\top$ is right-invertible. Thus, the result follows from Theorem~\ref{th-3.3}.
\end{proof}
\begin{thm}
 Let $\mathcal{C},~\mathcal{D}\subseteq \mathcal{P}_q(n)$ (with $q$ odd) be two subspace codes. Suppose the pair $\{\mathcal{C}, \mathcal{D}\}$ is an  LCP of subspace code. If $\lambda^2= -1$, then the pair $\{\{\mathcal{S}_\lambda(C_1, C_2^\perp)~|~C_1\in\mathcal{C}, C_2^\perp\in\mathcal{D}^\perp\}, \{\mathcal{S}_{\lambda}(C_1, C_2^\perp)^\perp~|~C_1\in\mathcal{C}, C_2^\perp\in\mathcal{D}^\perp\}\}$ is an LCP of subspace codes.
\end{thm}
\begin{proof}
  Let $\mathcal{G}$ denote a generator matrix for $\mathcal{S}_{\lambda}(C_1, C_2^\perp)$. Note that $\mathcal{S}_{\lambda}(C_1, C_2^\perp)^\perp$ is the dual of $\mathcal{S}_{\lambda}(C_1, C_2^\perp)$.
To ensure the desired LCP property, it is enough to verify that $\mathcal{G}\mathcal{G}^\top$ is right-invertible.
Let $G_1$ be a generator matrix of $C_1$ and let $H_2$ be a parity-check matrix of $C_2$. Then we may write
  \[
\mathcal{G}=\begin{pmatrix} 
G_{1} & \lambda G_1 \\ 
H_2 & -\lambda H_{2} 
\end{pmatrix} .
\]
Now, 
 \[
\mathcal{G} \mathcal{G}^\top=\begin{pmatrix} 
G_{1} & \lambda G_1 \\ 
H_2 & -\lambda H_{2} 
\end{pmatrix}  
\begin{pmatrix} 
G_{1}^\top & H_2^\top \\ 
\lambda G_1^\top & -\lambda H_{2}^\top 
\end{pmatrix}=\begin{pmatrix} 
 (1+\lambda^2)G_1G_{1}^\top & (1-\lambda^2)G_1H_{2}^\top\\ 
 (1-\lambda^2)H_{2}G_1^\top & (1+\lambda^2)H_2H_{2}^\top 
\end{pmatrix}.
\]
Since $C_1 \cap C_2 = \{0\}$, by Proposition~\ref{p-3} we find that both matrices $G_1H_2^\top$ and $H_{2}G_1^\top$ are right-invertible. 
Moreover, because $\lambda^2 = -1$, it follows from Theorem~\ref{th-3.3} that $\mathcal{G}\mathcal{G}^\top$ is right-invertible. Hence, the pair  $\{\{\mathcal{S}_\lambda(C_1, C_2^\perp)~|~C_1\in\mathcal{C}, C_2^\perp\in\mathcal{D}^\perp\}, \{\mathcal{S}_{\lambda}(C_1, C_2^\perp)^\perp~|~C_1\in\mathcal{C}, C_2^\perp\in\mathcal{D}^\perp\}\}$ forms an LCP of subspace codes, as desired.
\end{proof}
\begin{example}
Let $q = 5$ and choose $\lambda = 2 \in \mathbb{F}_5$. Since $2^2 = 4 \equiv -1 \pmod{5}$, we have
\[
x^6 - 1 = (x-1)(x+1)(x^2+x+1)(x^2-x+1),
\]
and
\[
x^6 - 4 = (x-2)(x-3)(x^2+2x+4)(x^2+3x+4).
\]
Define the sets
\[
\mathcal{C}=\{\langle (x+1)(x^2+x+1)\rangle, \langle (x-1)(x^2-x+1)\rangle\}
\]
and
\[
\mathcal{D}=\{\langle (x-2)(x^2+3x+4)\rangle, \langle (x-3)(x^2+2x+4)\rangle\}.
\]
It is straightforward to verify that the pair $\{\mathcal{C}, \mathcal{D}\}$ forms an LCP of subspace codes.

Consequently, the pair
$\{\{\mathcal{S}_2(C_1, C_2^\perp)~|~C_1\in\mathcal{C}, C_2^\perp\in\mathcal{D}^\perp\}, \{\mathcal{S}_{2}(C_1, C_2^\perp)^\perp~|~C_1\in\mathcal{C}, C_2^\perp\in\mathcal{D}^\perp\}\}$
is also an LCP of subspace codes.
\end{example}
\begin{thm}
 Let $\mathcal{C}\subseteq \mathcal{P}_q(n)$ (with $q$ odd) be a subspace code. Suppose $\mathcal{C}$ is an LCD subspace code. If $\lambda^2= 1$, then the pair $\{\{\mathcal{S}_\lambda(C_1, C_2)~|~C_1,C_2\in\mathcal{C}\}, \{\mathcal{S}_{\lambda}(C_1, C_2)^\perp~|~C_1, C_2\in\mathcal{C}\}\}$ forms an LCP of subspace codes.   
\end{thm}
\begin{proof}
  Let $\mathcal{G}$ denote a generator matrix for $\mathcal{S}_{\lambda}(C_1, C_2)$. Note that $\mathcal{S}_{\lambda}(C_1, C_2)^\perp$ is the dual of $\mathcal{S}_{\lambda}(C_1, C_2)$.
To ensure the desired LCP property, it is enough to verify that $\mathcal{G}\mathcal{G}^\top$ is right-invertible.
Let $G_i$ be a generator matrix of $C_i$ for $i=1,2$. Then we may write
  \[
\mathcal{G}=\begin{pmatrix} 
G_{1} & \lambda G_1 \\ 
G_2 & -\lambda G_{2} 
\end{pmatrix} .
\]
Now, 
 \[
\mathcal{G} \mathcal{G}^\top=\begin{pmatrix} 
G_{1} & \lambda G_1 \\ 
G_2 & -\lambda G_{2} 
\end{pmatrix}  
\begin{pmatrix} 
G_{1}^\top & G_2^\top \\ 
\lambda G_1^\top & -\lambda G_{2}^\top 
\end{pmatrix}=\begin{pmatrix} 
 (1+\lambda^2)G_1G_{1}^\top & (1-\lambda^2)G_1G_{2}^\top\\ 
 (1-\lambda^2)G_2G_{1}^\top & (1+\lambda^2)G_2G_{2}^\top 
\end{pmatrix}.
\]
Because $\mathcal{C}$ is LCD, so each $C_i\in\mathcal{C}$ satisfies $C_i \cap C^\perp_i = \{0\}$, and therefore $G_iG_i^\top$ is right-invertible, for $i=1,2$.
Since $q$ is odd and $\lambda^2 = 1$, we have $1+\lambda^2=2$, it follows from Proposition~\ref{p-3} that $\mathcal{G}\mathcal{G}^\top$ is right-invertible. Hence, the  considered pair forms an LCP of subspace codes, as desired.
\end{proof}
\begin{example}
Let $q = 5$ and choose $\lambda = 4 \in \mathbb{F}_q $. Observe that $4^2 = 16 \equiv 1 \pmod{5}$. The polynomial $x^{6} - 4$ factors over $\mathbb{F}_5$ as
\[
x^{6} - 4 = (x-2)(x-3)(x^{2}+2x+4)(x^{2}+2x+4).
\]
Consider the collection of subspaces
$\mathcal{C}=\{\langle x^{2}+x+1\rangle,\ \langle x^{2}-x+1\rangle\}.$
It is straightforward to verify that $\mathcal{C}$ forms an LCD subspace code. Moreover,
the pair $\{\{\mathcal{S}_4(C_1, C_2)~|~C_1,C_2\in\mathcal{C}\}, \{\mathcal{S}_{4}(C_1, C_2)^\perp~|~C_1, C_2\in\mathcal{C}\}\}$ 
constitutes an LCP of subspace codes.
\end{example}
\subsection{LCP of subspace codes derived from $k$-spread}
Let $k < n$. A $k$-spread of the vector space $\mathbb{F}_{q}^{n}$ is defined as a set of $k$-dimensional subspaces
$\{X_{1}, X_{2}, \ldots, X_{t}\}$ such that
\begin{enumerate}
\item[a)] each pair of distinct subspaces intersects only in the zero vector, i.e.,
$X_{i} \cap X_{j} = \{0\}, \text{ for all } i \neq j,$
\item[b)] the collection covers the entire space, i.e.,
$\bigcup_{i=1}^{t} X_{i} = \mathbb{F}_{q}^{n}.$
\end{enumerate}
A family satisfying these properties is called a $k$-spread of $\mathbb{F}_{q}^{n}$.
 \begin{thm}\cite[Theorem 5.7]{Hi98}\label{th-4.111}
   A $k$-spread of $\mathbb{F}_{q}^{n}$ exists if and only if $k$ is a divisor of $n$.  
 \end{thm}
\begin{thm}
 Let $\{U_{1}, U_{2}, \ldots, U_{t}\}$ be a k-spread of $\mathbb{F}_q^{n}$ with $n=2k$ and $t\geq 4$. If $\mathcal{C}_s=\{U_1,\ldots, U_s\}$ and  $\mathcal{D}_s=\{U_{s+1},\ldots, U_t\}$, for $2\leq s \leq t-2$, then the pair $\{\mathcal{C}_s, \mathcal{D}_s\}$ forms an LCP of subspace codes.  
\end{thm}
\begin{proof} 
By Theorem \ref{th-4.111}, a $k$-spread of $\mathbb{F}_q^{n}$ exists whenever $k\mid n$; in particular, for $n=2k$ we may fix a $k$-spread
\[
\{U_{1}, U_{2},\ldots, U_{t}\}.
\]
Choose any $1\le s \le t-1$, define
\[
\mathcal{C}_s=\{U_1,\ldots,U_s\} \text{ and }
\mathcal{D}_s=\{U_{s+1},\ldots,U_t\}.
\]
To show that $\{\mathcal{C}_s,\mathcal{D}_s\}$ is an LCP of subspace codes, it suffices to verify that
\[
C_i \cap D_j = \{0\} \quad \text{for all } C_i \in \mathcal{C}_s,  D_j \in \mathcal{D}_s.
\]
This follows immediately from the defining property of a $k$-spread. Hence, the proof is complete.
\end{proof}

\begin{example}
Let $q=5$, $n=6$, and $k=3$. Since $3 \mid 6$, so $3$-spread of $ \mathbb{F}_5^6$  exists. Moreover, it is well known that
\[
\mathbb{F}_5^6 \cong \mathbb{F}_{5^3} \times \mathbb{F}_{5^3}.
\] as $\mathbb{F}_5$-vector space.
For each $a \in \mathbb{F}_{5^3}$, define
\[
U_a = \{(x, ax) : x \in \mathbb{F}_{5^3}\} \text{ and }
U_\infty = \{(0, x) : x \in \mathbb{F}_{5^3}\}.
\]
Then each $U_a$ is a $3$-dimensional $\mathbb{F}_5$-subspace of $\mathbb{F}_5^6$, any two distinct subspaces intersect trivially, and every nonzero vector of $\mathbb{F}_5^6$ lies in exactly one of them. Hence
\[
\mathcal{S} = \{U_a : a \in \mathbb{F}_{5^3}\} \cup {U_\infty}
\]
is a $3$-spread of $\mathbb{F}_5^6$, consisting of $5^3 + 1 = 126$ subspaces. For each $2 \le s \le 124$, define
\[
\mathcal{C}_s = \{U_1, \ldots, U_s\} \text{ and }
\mathcal{D}_s = \{U_{s+1}, \ldots, U_{126}\}.
\]
Hence, for every $s$, the pair $\{\mathcal{C}_s, \mathcal{D}_s\}$ constitutes an LCP of subspace codes.
\end{example}
\begin{example}   
Consider the \( 2n \)-dimensional vector space  over \( \mathbb{F}_q \), i.e. 
\[
V = \{(x, y) : x, y \in \mathbb{F}^n\}.
\]
Let \( M_0, M_1, \dots, M_{q^n-1} \) be \( n \times n \) matrices over \( \mathbb{F} \) such that \( M_i - M_j \) is invertible for all \( i \neq j \). 

For each \( i = 0, 1, \dots, q^n-1 \), define the subspace  
\[
U_i = \{(x, x M_i) : x \in \mathbb{F}^n\},
\]  
which can be viewed as the graph of the linear map \( x \mapsto x M_i \).
Also define the subspace  
\[
U_{q^n} = \{(0, y) : y \in \mathbb{F}^n\}.
\]
Then the family  
\[
\mathcal{S} = \{U_0, U_1, \dots, U_{q^n}\}
\]  
forms a $k$-spread in \( V \) (see \cite[Theorem 3.4]{EM}). Note that \( U_i \)’s are \( n \)-dimensional subspaces. 
 For each $1 \le s \le q^n-2$, define
\[
\mathcal{C}_s = \{U_0, \ldots, U_s\} \text{ and }
\mathcal{D}_s = \{U_{s+1}, \ldots, U_{q^n}\}.
\]
Hence, for every $s$, the pair $\{\mathcal{C}_s, \mathcal{D}_s\}$ constitutes an LCP of subspace codes.   
\end{example}
\section{Applications to insertion error correction}\label{app-6}
In network coding, information is transmitted through subspaces. Let \(\{\mathcal{C}, \mathcal{D}\}\) be a linear complementary pair (LCP) of subspace codes. Specifically, suppose that \(\mathcal{C}\) and \(\mathcal{D}\) consist of  \(k\)-dimensional subspaces of \(\mathbb{F}_q^n\) such that $\{\mathcal{C}, \mathcal{D}\}$ forms a $k$-spread of $\mathbb{F}_q^n$. This complementary structure is crucial for error correction.

Assume a codeword \(C=\langle S\rangle \in \mathcal{C}\) is sent in a network and an insertion error \(E=\langle T \rangle\) (one-dimensional) occurs. The receiver then obtains the subspace  
\[
R =\langle S , T \rangle=C+E, \text{ at most a $(k+1)$- dimensional subspace}.
\]
\begin{itemize}
    \item \textbf{Detection:} Choose any \(D \in \mathcal{D}\). If \(R \cap D \neq \{0\}\), an insertion error is detected (since \(C \cap D = \{0\}\) for all \(C \in \mathcal{C}\) and \(D \in \mathcal{D}\)).
    \item \textbf{Correction:}
    \begin{itemize}
        \item[1.] If \(E \subseteq D\) for some \(D \in \mathcal{D}\), then \(R \cap D = E\). Removing \(E\) from \(R\) recovers \(C\).
        \item[2.] If \(E \not\subseteq D\) for every \(D \in \mathcal{D}\), find the unique \(C' \in \mathcal{C}\) such that \(C' \subseteq R\). Then \(C' = C\), and \(E\) is any complement of \(C\) in \(R\). The uniqueness of $C'$ follows from the property of $k$-spread of $\mathbb{F}_q^n$.
    \end{itemize}
\end{itemize}

Moreover, the structure of \(\mathcal{D}\) enables the determination of the error \(E\) and the recovery of the original codeword \(C\).
\begin{example}\label{ex-6.1} 
Let $q=2$, $k=2$, and $n=4$. 
 Let \(\mathbb{F}_2^4\) have basis \(\{e_1,e_2,e_3,e_4\}\). Consider the following 2‑dimensional subspaces of $\mathbb{F}_2^4$
\[
\begin{aligned}
U_1&=\langle e_1,e_2\rangle, &
U_2&=\langle e_3,e_4\rangle, &
U_3&=\langle e_1+e_3,\; e_2+e_4\rangle,\\
U_4&=\langle e_1+e_4,\; e_2+e_3+e_4\rangle, &
U_5&=\langle e_1+e_2+e_3,\; e_2+e_4\rangle .
\end{aligned}
\]
It is straightforward to verify that 
\[
\mathcal{S} = \{U_1, U_2, U_3, U_4, U_5\}
\]
forms a $2$-spread of $\mathbb{F}_2^{4}$. Further, we consider two subspace codes  
$\mathcal{C}=\{U_1,U_2,U_3\} \text{ and } 
\mathcal{D}=\{U_4,U_5\}.$ 
For every \(C\in\mathcal{C}\) and \(D\in\mathcal{D}\) we have \(C\cap D=\{0\}\). Thus \((\mathcal{C},\mathcal{D})\) is an LCP of subspace codes.
Suppose \(C=U_1\) is transmitted and an insertion error \(E=\langle e_1+e_4\rangle\) occurs, where \(E\subset U_4\). The received subspace is \(R=C\oplus E=\langle e_1,e_2,e_1+e_4\rangle\).
 Choose \(D=U_4\in\mathcal{D}\). Since \(C\cap D=\{0\}\) but \(R\cap D=E\neq\{0\}\), a nonzero intersection signals the presence of an insertion.
 Because \(E\subset D\), we have \(R\cap D=E\). Deleting the error space \(E\) from \(R\) yields \(C\), indeed \(\langle e_1,e_2\rangle\oplus E = R\).
Thus, identify \(E\) and recover the transmitted \(C\).   
\end{example}

\begin{example}
As in Example~\ref{ex-6.1}, consider the subspace codes
$\mathcal{C} = \{U_1, U_2, U_3\}, \text{ and } \mathcal{D} = \{U_4, U_5\},$
which form an LCP because \( C \cap D = \{0\} \) for every \(C \in \mathcal{C}\) and \(D \in \mathcal{D}\).
Suppose \(C = U_3\) is transmitted and an insertion error \(E = \langle e_1 \rangle\) occurs. Notice that \(E\) is not contained in \(U_4\) or \(U_5\).  
The received subspace is
$R = C \oplus E = \langle e_1+e_3,\; e_2+e_4,\; e_1 \rangle = \langle e_1,\, e_3,\, e_2+e_4 \rangle.$
 Take \(D = U_4 \in \mathcal{D}\). Since \(C \cap D = \{0\}\) but
$R \cap D = \langle e_2+e_3+e_4 \rangle \neq \{0\},$
the nonzero intersection detects an insertion.
 Here \(E \not\subseteq D\) for all \(D \in \mathcal{D}\).  
We find the unique codeword in \(\mathcal{C}\) contained in \(R\).  
Indeed \(U_3 = \langle e_1+e_3,\; e_2+e_4 \rangle \subseteq R\), while \(U_1,U_2 \not\subseteq R\).  
Thus \(C' = U_3 = C\).
\end{example}

\section{Conclusion}\label{con-7}
In this paper, we introduced the concept of LCPs of subspace codes and provided a characterization of such codes. We established necessary and sufficient conditions for the existence of LCPs in terms of generator matrices and subspace distances. Furthermore, we demonstrated the equivalence between LCPs of codes and their duals under suitable conditions. Several explicit constructions were presented, including techniques based on classical $[u|u+v]$-type constructions and $k$-spreads. These results extend the theory of complementary pairs of codes to the subspace coding framework and open new directions for applications in network coding and secure communications. In future work, it is of interest to consider nonempty fixed intersections of two subspace codes.
\section*{Acknowledgments}
The author extend their sincere gratitude to Dr. Satya Bagchi and Dr. Kuntal Deka for their meticulous proofreading of the manuscript and for offering invaluable insights and suggestions that greatly enhanced the quality of this work.

\end{document}